\documentclass[11pt]{article}

\usepackage[usenames,dvipsnames]{xcolor}
\usepackage[colorlinks,citecolor=blue,linkcolor=BrickRed]{hyperref}
\usepackage{makeidx} 
\usepackage{graphicx,tipa}
\usepackage{arcs,lmodern,fix-cm}
\usepackage{times}
\usepackage{amsfonts,latexsym,graphicx,epsfig,amssymb,color}
\usepackage{mathdots,amsmath,amsthm,amstext,setspace}
\usepackage{amsmath,amstext,amsthm,url,setspace, enumerate}
\usepackage{slashbox,multirow}
\usepackage{rotating}
\usepackage{verbatim}

\usepackage{float}
\usepackage{graphicx}   
\usepackage{xcolor}
\usepackage{microtype}
\usepackage{soul}
\usepackage{algorithm2e}
\usepackage{tabularx}
\usepackage{placeins}
\usepackage{mathtools}

\SetKwBlock{Repeat}{repeat}{end}

\newtheorem{theorem}{Theorem}[section]

\newtheorem{lemma}[theorem]{Lemma}

\theoremstyle{definition}

\newtheorem{definition}[theorem]{Definition}

\newcommand{\reals}{\mathbb{R}}
\newcommand{\dd}{\textrm{d}}

\newcommand{\ignore}[1]{}

\newcommand{\wec}[2]{$\textsc{TE}_d^{#1}\left(#2\right)$}
\newcommand{\we}[2]{$\textsc{ME}_d^{#1}\left(#2\right)$}

\newcommand{\ecb}[2]{$\textsc{EE}_{d}^{#1}\left(#2\right)$}

\newcommand{\naive}[1]{$\mathcal{N}^f_{#1}$}
\newcommand{\naivew}[1]{$\mathcal{N}_{#1}$}

\newcommand{\eenergyw}[1]{\mathcal E\left({#1}\right)}
\newcommand{\ttimew}[1]{\mathcal T \left( {#1} \right) }

\newcommand{\nlpw}[2]{\textsc{NLP}_{#1}^{#2} }
\newcommand{\nlpwbar}[1]{\textsc{NLP}_{#1}'}
\newcommand{\nlpwbarbar}[1]{\textsc{NLP}_{#1}''}

\begin{document}
\title{\bf 
Time-Energy Tradeoffs \\
for Evacuation by Two Robots \\
in the Wireless Model
\footnote{This is the full version of the paper with the same title which will appear in the proceedings of the 
26th International Colloquium on Structural Information and Communication Complexity (SIROCCO'19) L'Aquila, Italy during July 1-4, 2019.}
}

\author{
Jurek Czyzowicz\footnotemark[1] 
\and
Konstantinos Georgiou\footnotemark[2] 
\and
Ryan Killick\footnotemark[4] 
\and 
Evangelos Kranakis\footnotemark[4] 
\and
Danny Krizanc\footnotemark[7]
\and
Manuel Lafond\footnotemark[5] 
\and
Lata Narayanan\footnotemark[8] 
\and
Jaroslav Opatrny\footnotemark[8]
\and
Sunil Shende\footnotemark[9]
}

\def\thefootnote{\fnsymbol{footnote}}
\footnotetext[1]{
Universite du Québec en Outaouais, Gatineau, Qu\'{e}bec, Canada, \texttt{jurek.czyzowicz@uqo.ca}
}
\footnotetext[2]{
Dept. of Mathematics, 
Ryerson University, 
Toronto, ON, Canada, \texttt{konstantinos@ryerson.ca}
}
\footnotetext[4]{
School of Computer Science, Carleton University, Ottawa ON, Canada, \texttt{ryankillick,kranakis@scs.carleton.ca}
}
\footnotetext[7]{
Department of Mathematics \& Comp. Sci., Wesleyan University, Middletown, CT, USA, \texttt{dkrizanc@wesleyan.edu}
}
\footnotetext[8]{
Department of Comp. Sci. and Software Eng., Concordia University, Montreal, Qu\'{e}bec,  Canada, \texttt{lata,opatrny@encs.concordia.ca}
}
\footnotetext[9]{
Department of Computer Science, Rutgers University, Camden, USA,
\texttt{shende@camden.rutgers.edu}
}
\footnotetext[5]{
Department of Computer Science, Universit\'{e} de Sherbrooke, Qu\'{e}bec,  Canada
\texttt{Manuel.Lafond@usherbrooke.ca}
}
\maketitle

\begin{abstract}
Two robots stand at the origin of the infinite line and are tasked with searching collaboratively for an exit at an unknown location on the line. They can travel at maximum speed $b$ and can change speed or direction at any time. The two robots can communicate with each other at any distance and at any time. The task is completed when the last robot arrives at the exit and evacuates. We study time-energy tradeoffs for the above evacuation problem. The evacuation time is the time it takes the last robot to reach the exit. The energy it takes for a robot to travel a distance $x$ at speed $s$ is measured as $xs^2$. The total and makespan evacuation energies are respectively the sum and maximum of the energy consumption of the two robots while executing the evacuation algorithm. 

Assuming that the maximum speed is $b$, and the evacuation time is at most $cd$, where $d$ is the distance of the exit from the origin, we study the problem of minimizing the total energy consumption of the robots. We prove that the problem is solvable only for $bc \geq 3$.  For the case $bc=3$, we give an optimal algorithm, and give upper bounds on the energy for the case $bc>3$. 

We also consider the problem of minimizing the evacuation time when the available energy is bounded by $\Delta$. Surprisingly, when  $\Delta$ is a constant, independent of the distance $d$ of the exit from the origin, we prove that  evacuation is possible in time $O(d^{3/2}\log d)$, and this is optimal up to a logarithmic factor. When $\Delta$ is linear in $d$, we give upper bounds on the evacuation time. 

%
%

\vspace{0.5cm}
\end{abstract}



\section{Introduction}

Linear search is an online problem in which a robot is  tasked with finding an exit placed at an unknown location on an infinite line. It has long been known that the classic doubling strategy, which guarantees a search time of $9d$ for an exit at distance $d$ from the initial location is optimal for a robot travelling at speed at most 1 (see any of the books~\cite{ahlswede1987search,alpern2003theory,stone1975theory} for additional variants, details and information). If even one more robot is allotted to the search then clearly an exit at distance $d$ can always be found in time $d$ by one of the robots. Therefore the problem of group search by multiple robots on the line is concerned with minimizing the time the {\em last} robot arrives at the exit; the problem is also called {\em evacuation}. It was first introduced as part of a study on cycle-search~\cite{CGGKMP} and further elaborated on an infinite line for multiple communicating robots with crash~\cite{PODC16} and Byzantine faults~\cite{isaacCzyzowiczGKKNOS16}.

The time taken for group search on the line clearly depends on the communication capabilities of the robots. In the wireless communication model, the robots can communicate at any time and over any distance. In the face-to-face communication model, the robots can only communicate when they are in the same place at the same time. A straightforward algorithm achieves evacuation time $3d$ in the wireless model, and can be seen to be optimal, while it has been shown that in the face-to-face model, two robots cannot achieve better evacuation time than one robot \cite{chrobak2015group}. 

In this paper, we consider the {\em energy} required for group search on the line. We use the energy model proposed in \cite{CGKKKLNOS19wireless} in which the energy consumption of a robot travelling a distance $x$ at speed $s$ is proportional to $xs^2$. This model is motivated by the concept of viscous drag in fluid dynamics; see Section~\ref{sec:prelims} for more details.  The authors of  ~\cite{CGKKKLNOS19wireless}, 
studied the question of the minimum energy required for group search on the line by two robots travelling at speed at most $b$ while guaranteeing that both robots reach the exit within time $cd$, where $d$ is the distance of the exit from the starting position of the robots. For the special case $b=1, c=9$, they proved the surprising result that two robots can evacuate with less energy than one robot, while taking the same evacuation time.

Our main approach throughout the paper is to investigate time-energy tradeoffs for group search by two robots in the wireless communication model. Assuming that the maximum speed is $b$, and the evacuation time is at most $cd$, where $d$ is the distance of the exit from the origin, we study the problem of minimizing the total energy consumption of the robots. 
We also consider the problem of minimizing the evacuation time when the available energy is bounded by $\Delta$.

\subsection{Model and problem definitions}
\label{sec:prelims}


Two robots are placed at the origin of an infinite line. An exit is located at unknown distance $d$ from the origin and can be found if and only if a robot walks over it.  A robot can change its direction or speed at any time, e.g., as a function of its distance from the origin, or the distance walked so far. Robots operate under the wireless model of communication in which messages can be transmitted between robots instantaneously at any distance.
Feasible solutions are robots' trajectories in which, eventually, both robots evacuate, i.e. they both reach the exit. Given a location of the exit, the time by which the second robot reaches the exit is referred to as the \textit{evacuation time}. We distinguish between constant-memory robots that can
only travel at a constant number of hard-wired speeds, and unbounded-memory robots that can dynamically compute speeds and distances, and travel at any possible speed.



The energy model being used throughout the paper is motivated from the concept of viscous drag in fluid dynamics \cite{batchelor2000}.  In particular, an object moving with constant speed $s$ will experience a {\em drag} force $F_{D}$ proportional\footnote{The constant of proportionality has (SI) units $kg/m$ and depends, among other things, on the shape of the object and the density of the fluid through which it moves.} to $s^2$. In order to maintain the speed $s$ over a distance $x$ the object must do work equal to the product of $F_D$ and $x$ resulting in a continuous energy loss proportional to the product of the object's squared speed and travel distance. For simplicity we take the proportionality constant to be one, and define the energy consumption moving at constant speed $s$ over a segment of length $x$ to be $xs^2$. We extend the definition of energy for a robot moving in the same direction from point $a$ to point $b$ on the line, using speed $s(x)\in \reals, x\in [a,b]$, as $\int_a^b s^2(x) \dd x$. 
The total energy of a specific robot traversing more intervals, possibly in different directions, is defined as the sum of the energies used in each interval.

Given a collection of robots, the \textit{total evacuation energy} is defined as the sum of the robots' energies used till both robots evacuate. 
Similarly, we define the \textit{makespan evacuation energy} as the maximum energy used by any of the two robots.

%
%
%
%
For each $d>0$ there are two possible locations for the exit to be at distance $d$ from the origin: we will refer to either of these as input instances $d$ for the group search problem. More specifically, we are interested in the following three optimization problems:

\begin{definition}
Problem \ecb{b}{c}: 
Minimize the  total evacuation energy, given that the  evacuation time is no more than $cd$ (for all instances $d$) and using speeds no more than $b$. 
\end{definition}




\begin{definition}
Problem \wec{b}{\Delta}: Minimize the  evacuation time, given that the total evacuation energy is no more than $\Delta$ (for all instances $d$), and using speeds at most $b$.
\end{definition}

\begin{definition}
Problem \we{b}{\Delta}: Minimize the  evacuation time, given that the makespan evacuation energy is no more than $\Delta$ (for all instances $d$), and using speeds at most $b$. 
\end{definition}

For the last two problems, we consider two cases when the evacuation energy $\Delta$ is a constant and when it is linear in $d$. 




\subsection{Our results }


Consider the following intuitive and simple algorithm for wireless evacuation, which is a parametrized version of a well-known algorithm for the case of unit speed robots that achieve evacuation time $3d$.
\begin{definition}[Algorithm Simple Wireless Search \naivew{s,r}]
\label{def: naive wireless algo}
Robots move at opposite directions with speed $s$ until the exit is found. The finder announces ``exit found'' and halts. The other robot changes direction and moves at speed $r$ until the exit is reached. 
\end{definition}

We analyze the behaviour of this algorithm for all three proposed problems, and determine the speeds that achieve the minimum evacuation energy (or time) among all algorithms of this class, while respecting the given bound on evacuation time (resp.  energy). In some cases, the algorithms derived are shown to be optimal. In particular, our main results are the following:

\begin{enumerate}

\item We show that the problem \ecb{b}{c} admits a solution if and only if $cb\geq 3$. Furthermore, for every $c, b > 0$ with $cb=3$, we show that  the optimal  total evacuation energy is $4b^2d$, and this is achieved by \naivew{s,r} with $s=r=b$ (Theorem~\ref{lower-bound}). 

\item For every $c, b> 0$ with $cb \geq 3$, we derive the optimal values of $s$ and $r$ for the algorithm \naivew{s,r} that minimize the total evacuation energy (Theorem~\ref{thm: optimal speeds naive bounded speeds wireless}). 

\item We observe that if total or makespan energy $\Delta$ is a constant,  problems \wec{b}{\Delta}  and \we{b}{\Delta}  cannot be solved by robots that can only use a finite number of speeds. We prove that if $\Delta$ is bounded by a constant, the optimal evacuation time is  $\Omega(d^{3/2})$  (see Theorem~\ref{thm: lower bound constant energy}). Somewhat surprisingly, we give an algorithm with total evacuation time $O(d^{3/2} \log d)$ (see Theorem~\ref{thm: upper bound constant max energy e<=1}); thus the algorithm is optimal up to a logarithmic factor. Our algorithm requires the robots to continuously change their speed at every distance $x$ from the origin. This is the only part that requires robots
  to have unbounded memory. 

\item For the problems \wec{b}{\Delta}  and \we{b}{\Delta}  with total or makespan energy  $\Delta= O(d)$ and $b=1$, we give upper bounds on the  total evacuation time (see Theorems~\ref{thm: upper bound constant max energy e<=1} and ~\ref{thm: optimal speeds naive wec makespan energy bound} respectively). 
\end{enumerate}
All proofs missing from the main text can be found in the appendix.

\subsection{Related Work} 

In group search, a set of communicating robots interact and co-operate by exchanging information in order to complete the task which usually involves finding an exit placed at an unknown location within a given search domain. Some of the pioneering results related to our work are concerned with search on an infinite domain, like a straight line~\cite{baezayates1993searching,beck1964linear,bellman1963optimal,kao1996searching},  while others with search on the perimeter of a closed domain like unit disk~\cite{CGGKMP} or equilateral triangle or square~\cite{CKKNOS}. The communication model being used may be  either wireless~\cite{CGGKMP} or F2F~\cite{Watten2017,CGKNOV,CKKNOS}. Search and evacuation problems with a combinatorial flavour have been recently considered in \cite{CGKKKNOS18a,CGKKKNOS18b} and search-and-fetch problems in~\cite{GeorgiouKK16,GeorgiouKK17}, while \cite{CGS18} studied average-case/worst-case trade-offs for a specific evacuation problem on the disk. The interested reader may also wish to consult a recent survey paper~\cite{CGK19search} on selected search and evacuation topics.

Traditional approaches to evaluating the performance of search have been mostly concerned with time. This is apparent in the book~\cite{alpern2003theory} and the research described in the seminal works on deterministic~\cite{baezayates1993searching}, stochastic~\cite{beck1964linear,bellman1963optimal} and randomized~\cite{kao1996searching} search and continued up to the most recent research papers on linear search for robots with terrain dependent speeds~~\cite{CzyzowiczKKNOS17} and robots with Byzantine~\cite{isaacCzyzowiczGKKNOS16} and crash fault behaviour~\cite{PODC16} (see also the survey paper~\cite{CGK19search}). Aside from the research by~\cite{demaine2006online}, in which the authors are looking at the turn cost when robots change direction during the search,  little or no research has been conducted on other measures of performance.

The first paper on search and evacuation to change this focus from optimizing the time to the energy consumption required to find the exit as well as to time/energy tradeoffs is due to~\cite{CGKKKLNOS19wireless}. The authors determine optimal (and in some cases nearly optimal) linear search algorithms inducing the lowest possible energy consumption and also propose a linear search algorithm that simultaneously achieves search time $9d$ and consumes energy $8.42588d$, for an exit located at distance $d$ unknown to the robots. However, the previously mentioned paper~\cite{CGKKKLNOS19wireless} differs from our present work in that the authors focus exclusively on the face-to-face communication model while here we focus on the wireless model. In the present paper, we extend the results of~\cite{CGKKKLNOS19wireless} to the realm of the wireless communication model and study time/energy trade-offs for evacuating two robots on the infinite line. Despite their apparent similarities, the face-to-face and wireless communication models lead to completely different approaches for the design of efficient linear search algorithms.

%
%

\section{Minimizing Energy Given Bounds on Evacuation Time and Speed}
\label{sec: min energy, given bounds on time, wireless}

This section is devoted to the problem \ecb{b}{c} of minimizing the total evacuation energy, given that the robots can travel at speed at most $b$ and are required to complete the evacuation within time $cd$ for every instance $d$ where $d$ is the distance of the exit from the origin. We start with establishing a necessary condition on the product $bc$.

\begin{lemma}
\label{No online:lm}
No online (wireless) algorithm can solve \ecb{b}{c} if $bc<3$.
\end{lemma}

Next we show that algorithm \naivew{b,b} is an optimal solution to the problem \ecb{b}{c} when $bc =3$. We start with the following lemma:

\begin{lemma}
\label{speed-b}
Let $b, c > 0$ with $bc = 3$ and consider an evacuation algorithm such that robots use maximum speed $b$ and evacuate by time $cd$ for an exit at distance $d$ from the origin. 
Then for every $d> 0$, the points $d, -d$, must be visited at time $d/b$.
\end{lemma}

\begin{theorem}
\label{lower-bound}
For every $b, c > 0$ with $bc=3$, the algorithm \naivew{b, b} is the only feasible solution to \ecb{b}{c}, and is therefore optimal, and has  total energy consumption $4b^2d$.
\end{theorem}

\begin{proof} (Theorem~\ref{lower-bound})
Lemma~\ref{speed-b}  implies that in order to achieve an evacuation time $cd$, both robots must use the maximum speed $b$ and explore in different directions. If the exit is found at distance $d$ by one of the robots, the time is $d/b$, and therefore, the other robot must travel at the maximum speed $b$ in order to arrive at the exit in time $cd$. Thus, the only algorithm that can evacuate within time $cd$ while using speed at most $b$ is \naivew{b, b}.  A total distance of $4d$ is  travelled by the two robots, all at speed $b$,  therefore the total energy consumed is $4b^2d$.
\end{proof}

%
%

Next we consider the case of $c, b >3$ and determine the optimal choices of speeds $s,r$ for \naivew{s,r}, as well as the induced total evacuation energy and competitive ratio for problem \ecb{b}{c}. 

\begin{theorem}
\label{thm: optimal speeds naive bounded speeds wireless}
Let $\delta=2+\sqrt[3]{2}\approx 3.25992$. 
For every $c,b>0$, problem \ecb{b}{c} admits a solution by algorithm \naivew{s,r}\ if and only if $cb\geq 3$. 
 For the spectrum of $c,b$ for which a solution exists, the following choices of speeds $s,r$ are feasible and optimal for \naivew{s,r}
$$
\begin{tabular}{l|cc}
			& $3 \leq cb \leq \delta$		&	 $cb > \delta $\\
\hline
$s$ 	& 	$\frac{b}{b c-2}$	& 	$\frac{\delta}{\sqrt[3]{2} c} $ \\
$r$	& $b$	& $\frac{\delta}{ c} $		 \\
\end{tabular}
$$
The induced total  evacuation energy is $f(cb)\frac{2d}{c^2}$, where
$$
f(x):=
\left\{
\begin{array}{ll}
\frac{x^2}{(x-2)^2}+x^2 &, 3 \leq x \leq \delta \\
\frac{1}{2} \left(2+\sqrt[3]{2}\right)^3 &, x >\delta
\end{array}
\right.
$$
\end{theorem}

It was observed in~\cite{CGKKKLNOS19wireless} that the optimal offline solution, given that $d$ is known, equals $\frac{2d}{c^2}$.  The competitive ratio  is
given by $
\sup_{d} \frac{c^2}{2d}~ e(c,b,d) = f(cb)
$
for algorithms inducing total evacuation energy $e(c,b,d)$. The competitive ratio of \naivew{s,r} for the choices of Theorem~\ref{thm: optimal speeds naive bounded speeds wireless} is summarized in Figure~\ref{fig: compratiowireless}.
\begin{figure}[h!]
  \centering
  \includegraphics[width=0.5\linewidth]{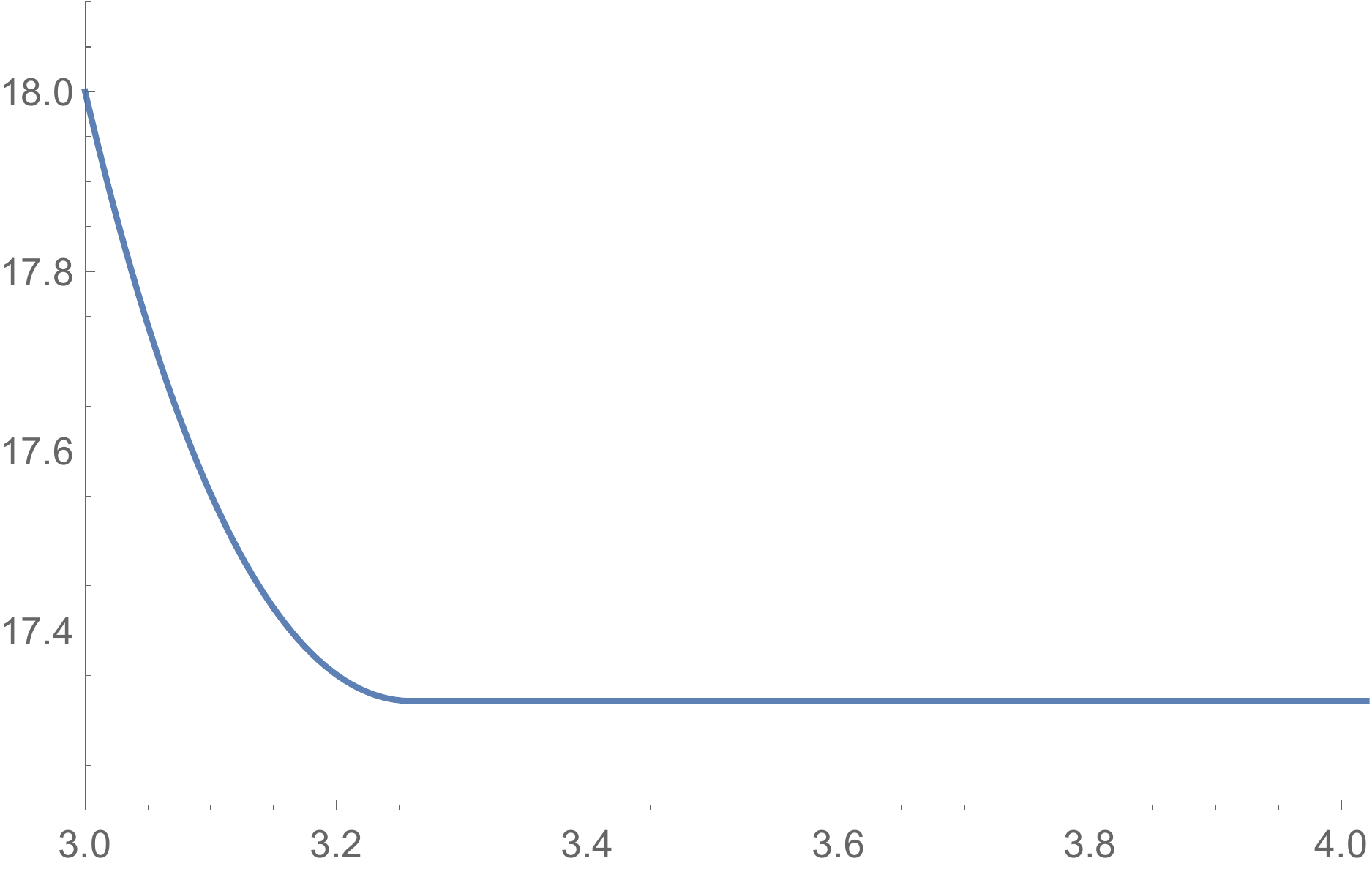}
\caption{
The competitive ratio of \naivew{s,r} for the choices of Theorem~\ref{thm: optimal speeds naive bounded speeds wireless}.
}
\label{fig: compratiowireless}
\end{figure}
Note that in particular, Theorem~\ref{thm: optimal speeds naive bounded speeds wireless} claims that the competitive ratio only depends on the product $cb$, and when $cb=3$, the competitive ratio is 18 and is decreasing in $cb$ (strictly only when $cb<\delta$). The optimal speed choices for the unbounded problem \ecb{c}{\infty}\ are exactly those that appear under case $cb>\delta$. 
The remaining of the section is devoted to proving Theorem~\ref{thm: optimal speeds naive bounded speeds wireless}.

First we derive closed formulas for the performance of \naivew{s,r}. 
From the definition of energy used, and given that the robots move at speed 1, we deduce what the evacuation time and energy are when the exit is placed at distance $d$ from the origin.
The following two functions will be invoked throughout our argument below. 
\begin{align}
\ttimew{s,r}&:= \frac1s+\frac2r \label{equa: ttimew}\\
\eenergyw{s,r}&:= s^2 + r^2	\label{equa: eenergyw}
\end{align}
\begin{lemma}
\label{lem: energy time formulas wireless}
Let $b,c$ be such that there exist $s,r$ for which \naivew{s,r} is feasible. 
Then, for instance $d$ of \ecb{b}{c}, the induced evacuation time of \naivew{s,r} is 
$
d\cdot \ttimew{s,r}
$
and the induced total evacuation energy is 
$
2d\cdot \eenergyw{s,r}
.$
\end{lemma}
Next we show the spectrum of $c,b$ for which \naivew{s,r}\ is applicable. 
\begin{lemma}
\label{lem: naive as NLP wireless}
Algorithm \naivew{s,r} gives rise to a feasible solution to problem \ecb{b}{c} if and only if $bc\geq 3$.
For every such $b,c>0$, the optimal choices  of \naive{s,r} can be obtained by solving Convex Program:
\begin{align}
& ~~\min_{s,r \in \reals} \eenergyw{s,r} \tag{$\nlpw{c}{b}$}\\
s.t. ~~&
\ttimew{s,r} \leq c  \notag \\
&0 \leq s,r\leq b. \notag
\end{align}
Moreover, if $s_0, r_0$ are the optimizers to $\nlpw{c}{b}$, then the competitive ratio of \naivew{s_0,r_0} equals 
$
c^2 \cdot \eenergyw{s_0,r_0}.
$
\end{lemma}

A corollary of Lemma~\ref{lem: naive as NLP wireless} is that any candidate optimizer to $\nlpw{c}{b}$ satisfying 1st order necessary optimality conditions is also a global optimizer. As a result, the proof of Theorem~\ref{thm: optimal speeds naive bounded speeds wireless} follows by showing the proposed solution is feasible and satisfies 1st order necessary optimality conditions. This is done in Lemmata~\ref{lem: opt wireless 3<=cb<=3.25} and~\ref{lem: opt wireless cb>3.25}.

Towards proving that 1st order optimality conditions are satisfied, we argue first that for all $c,b>0$ with $cb\geq 3$, the optimizers of $\nlpw{c}{b}$ satisfy the time constraint tightly. Indeed, if not, then one could reduce any of the values among $s,r$ to make the constraint tight, improving the induced energy. Hence, in the optimal solutions to $\nlpw{c}{b}$, any of $s,r\leq b$ could be additionally tight or not. 
In what follows, $\delta$ represents $2+\sqrt[3]{2}$, as in the statement of Theorem~\ref{thm: optimal speeds naive bounded speeds wireless}.

\begin{lemma}
\label{lem: opt wireless 3<=cb<=3.25}
For each $c,b>0$ for which $3\leq cb \leq \delta$, the optimal solution to $\nlpw{c}{b}$ is given by 
$s=\frac{b}{b c-2}, r=b$. 
\end{lemma}

\begin{lemma}
\label{lem: opt wireless cb>3.25}
For each $c,b>0$ for which $cb > \delta$, the optimal solution to $\nlpw{c}{b}$ is given by 
$s=\frac{\delta}{\sqrt[3]{2} c}, r=\frac{\delta}{ c}$. 
\end{lemma}

\begin{proof} (Theorem~\ref{thm: optimal speeds naive bounded speeds wireless})
By Lemma~\ref{lem: naive as NLP wireless} and  Lemma~\ref{lem: opt wireless 3<=cb<=3.25}, the optimal induced energy when $3\leq cb \leq \delta$ is 
$$
2d \eenergyw{\frac{b}{b c-2},b}=
2d\left( \frac{b^2}{(b c-2)^2}+b^2\right)
$$
and the induced competitive ratio is 
$$
(cb)^2\left( 1+ \frac{1}{(cb-2)^2}\right).
$$
Finally, by Lemma~\ref{lem: naive as NLP wireless} and Lemma~\ref{lem: opt wireless cb>3.25}, the optimal induced energy when $ cb > \delta$ is 
$$
2d \eenergyw{\frac{1+2^{2/3}}{c}, \frac{2+\sqrt[3]{2}}{c}}
=
d \frac{ \left(2+\sqrt[3]{2}\right)^3}{c^2}.
$$
Hence the competitive ratio is constant and equals 
$$
\frac{1}{2} \left(2+\sqrt[3]{2}\right)^3  
\approx 17.3217,
$$
completing the proof of Theorem~\ref{thm: optimal speeds naive bounded speeds wireless}.
\end{proof}

\section{Minimizing Evacuation Time, Given Constant Evacuation Energy}
\label{sec: min time, given constant energy}

In this section we consider the problem of minimizing evacuation time, given {\em constant}  total (or makespan) evacuation energy.  
First we observe that if the robots can use only a finite number of speeds, there is no feasible solution to the problems \we{b}{\Delta} or \wec{b}{\Delta}.

\begin{theorem}
\label{If Delta is a constant:thm}
If $\Delta$ is a constant,  and the robots have access to only a finite number of speeds, there is no feasible solution to the problems \we{b}{\Delta} or \wec{b}{\Delta} 
 \end{theorem}
\begin{proof} (Theorem~\ref{If Delta is a constant:thm})
Suppose the robots can only use speeds in a finite set. Wlog let $s$ be the minimum speed in the set. Define $d' = \Delta/s^2$, and place the exit at $d'+\epsilon$ for any $\epsilon>0$.
Travelling at any speed at or above $s$, it is impossible for even one of the robots to reach the exit with energy $\leq \Delta$. 
\end{proof}

Next we prove a lower bound on the evacuation time in this setting. 

\begin{theorem}
\label{thm: lower bound constant energy}
For every constant $e\in \reals_+$, the optimal evacuation time for problem \we{b}{e} is  $\Omega(d^{3/2})$, asymptotically in $d$. 
\end{theorem}

\begin{proof} (Theorem~\ref{thm: lower bound constant energy})
For any arbitrarily large value of $d$, we place the exit at distance $d$ from the origin. For any robot to reach the exit before running out of battery, a robot can travel at speed at most $e/\sqrt{d}$. Therefore the time for even the first robot to reach the exit is at least $\frac{d}{e/\sqrt{d}} = d^{3/2}/e$.
\end{proof}

Note that the above lower bound also holds for problem \wec{b}{e} (if the total evacuation energy is no more than $e$, then also the makespan evacuation energy is no more than $e$). Next we prove that this naive lower bound is nearly tight (up to a $\log d$ factor). 
First we consider the case that $e\leq 1$. Then, we show how to modify our solution to also solve the problem when $e>1$. 

The key idea is to allow functional speed $s=s(x)$ to depend on the distance $x$ of the robot from the origin. We will make sure that the choice of $s$ is such that, for every large enough $d$, once the exit is located at distance $d$, there is ``enough'' leftover energy for the other robot to evacuate too. For that, we will choose the maximum possible speed $r$ (which can now depend on $d$, and which will be constant) so as to evacuate without exceeding the maximum energy bounds. Notably, even though our algorithmic solution is described as a solution to \wec{b}{e}, it will be transparent in the proof that it is also feasible to \we{b}{e}.

\begin{theorem}
\label{thm: upper bound constant max energy e<=1}
For every constant $e\leq 1$, problem \wec{b}{e} admits a solution by \naivew{s,r}, where (functional) speed $s$ is  chosen as 
$$
s(x) = \frac{1}{\sqrt{2+2x} \left( 1/e+\log(1+x) \right) }.
$$
When the exit is found (hence its distance $d$ from the origin becomes known), speed $r$ is chosen as 
$$
r=
\sqrt{\frac{e}{2d\left( e \log (d+1)+1\right)}},
$$
inducing evacuation time $O\left(d^{3/2}\log d\right)$, where in particular the constant in the asymptotic (in $d$) is independent of $e$. 
\end{theorem}

\begin{proof} (Theorem~\ref{thm: upper bound constant max energy e<=1})
First we observe that since $e\leq 1$, $s(x)\leq 1$ for all $x\geq 0$. Given that $d$ is at least, say, 1, it is also immediate that $r\leq 1$, hence the speed choices comply with the speed bound. 

The exit placed at distance $d$ from the origin is located by the finder in time 
$$
\int_0^d \tfrac{1}{s(x)} \dd x
=
\frac{2 \sqrt{2} \left((d+1)^{3/2} (3 e \log (d+1)-2 e+3)+2 e-3\right)}{9 e}
\leq d^{3/2} \log d,
$$
where the inequality holds for every $e\leq 1$, and for big enough $d$. 

When the exit is located by a robot, the other robot is at distance $2d$ from the exit. Moreover, each of the robots have used energy 
$$
\int_0^d s^2(x) \dd x =\frac{e}{2}-\frac{e}{2 e \log (d+1)+2},
$$
hence the leftover energy for the non-finder (i.e., the robot that did not find the exit) to evacuate is at least 
$$
e-2\left( \frac{e}{2}-\frac{e}{2 e \log (d+1)+2} \right) 
=\frac{e}{e \log (d+1)+1}.
$$
The non-finder is informed of $d$, and hence can choose constant speed $r$ so as to use exactly all of the leftover energy, i.e. by choosing $r$ satisfying 
$$\int_0^{2d} r^2 \dd x = \frac{e}{e \log (d+1)+1}.$$
Note that our choice of $r$ is also feasible to problem \we{b}{e}. 
Solving for $r$ gives the value declared at the statement of the theorem. Finally, choosing this specific value of $r$, the non-finder needs additional $2d/r$ time to evacuate, which is at most 
$$
(2d)^{3/2}
\sqrt{\frac{\left( e \log (d+1)+1\right)}{e}}
\leq 
(2d)^{3/2}
\sqrt{\frac{ \log (d+1)}{e}}
\leq d^{3/2} \log d, 
$$
where the last inequality holds for big enough $d$, since $e$ is constant. So the overall evacuation time is no more than $2d^{3/2} \log d$, for big enough $d$, as promised. 
\end{proof}

It remains to address the case $e> 1$.
For this, we recall that we solve \wec{b}{e} for large enough values of $d$, and we modify our solution so as to choose functional speed 
$$
\bar s(x) := \min\{ s(x), 1\},
$$
effectively using even less energy than before. The distance that is traversed at speed 1 depends only on constant $e$, and hence the additional evacuation time is $O(1)$ with respect to $d$.

\section{Minimizing Evacuation Time with Bounded Linear Total Evacuation Energy}
\label{sec: min time, given linear energy}

In this section we study the problem \wec{1}{\Delta} of minimizing the total evacuation time, where $\Delta = ed$ for some constant $e$. We show how to choose optimal speed values $s,r$ for algorithm \naivew{s,r}. Note that even though $d$ is unknown to the algorithm, speeds $s,r$ may depend on the known constant $e$, and the maximum speed $b=1$.

In this section we prove the following theorem:

\begin{theorem}
\label{thm: optimal speeds naive wec}
Let $\delta=2+\sqrt[3]{2}\approx 3.25992$. 
For every constant $e\in \reals_+$, problem \wec{1}{ed} admits a solution by \naivew{s,r}, where speeds $s,r$ are chosen as follows
$$
\begin{tabular}{l|ccc}
			& 	$e< \delta$ 	&	$e\in [\delta,4)$ & $e \geq 4$   \\
\hline
$s$ 		& $\sqrt{\tfrac{e}{2\left(1+2^{2/3}\right)}}$		& 	$\sqrt{\tfrac{e-2}2}$	&	$1$ \\
$r$		&	$\sqrt{\tfrac{e}{\left(2+2^{1/3}\right)}}$	& 		$1$							&	$1$	 \\
\end{tabular}
$$
The induced total evacuation time  is given by $g(e)d$ where $g(e)$ is given by:
$$
g(e):=
\left\{
\begin{array}{ll}
\sqrt{  \tfrac{\left(2+2^{1/3}\right)^3}{e}   }							&, e < \delta \\	
2+\sqrt{\tfrac{2}{e-2}}							&, e\in [\delta,4) \\
3							&, e\geq 4 \\	
\end{array}
\right.
$$
\end{theorem}

First we observe that, given the values of $s=s(e), r=r(e)$, it is a matter of straightforward calculations to verify, assuming they are feasible and optimal, that the induced  evacuation time is indeed equal to $g(e)d$ as promised. 
Given Lemma~\ref{lem: energy time formulas wireless}, we know that the optimal speed choices for algorithm~\naivew{s,r}, for problem \wec{1}{ed} are obtained as the solution to the following NLP. 
\begin{align}
& ~~\min_{s,r \in \reals} \tfrac1s+\tfrac2r \tag{$\nlpwbar{e}$}\\
s.t. ~~&
2(s^2+r^2) \leq e  \notag \\
&0 \leq s,r\leq 1 \notag
\end{align}
The optimal solutions to $\nlpwbar{e}$ can be obtained by solving complicated algebraic systems and by invoking KKT conditions, for the various values of $e$, as we also did for $\nlpw{c}{b}$. However, the advantage is that one can map the optimal solutions to $\nlpw{c}{1}$, see Theorem~\ref{thm: optimal speeds naive bounded speeds wireless} and use $b=1$, to feasible solutions to $\nlpwbar{e}$. Then, we just need to \textit{verify} 1st order optimality conditions for the candidate optimizers. Since the NLP is convex, these should also be unique global optimizers. 

Indeed, one of the critical structural properties pertaining to the optimizers of $\nlpw{c}{1}$ is that the time constraint $\tfrac1s+\tfrac2r\leq cd$ is satisfied \textit{tightly}. At the same time, the optimal speed values, as described in Theorem~\ref{thm: optimal speeds naive bounded speeds wireless}, as a function of $c$, achieve 
evacuation energy equal to $f(c)d\frac{2}{c^2}$. 
Attempting to find the correspondence between parameters $c,e$ (and problems $\nlpw{c}{1}$, $\nlpwbar{e}$), we consider the transformation $f(c)\frac{2}{c^2}=e$. For the various cases of the piece-wise function $f$, the transformation gives rise to the piece-wise function $g$ and optimal speeds $s,r$ (as a function of $e$) of Theorem~\ref{thm: optimal speeds naive wec}. 

Overall, the previous approach provides just a mapping between the provable optimizers $s(c), r(c)$ to $\nlpw{c}{1}$, and candidate solutions $s(e), r(e)$ to $\nlpwbar{e}$, and more importantly, it saves us from solving complicated algebraic systems induced by KKT conditions. What we verify next (which is much easier), is that feasibility and KKT conditions are indeed satisfied for the obtained candidate solutions $s(e), r(e)$.
 Since the NLP is convex, that also shows that $s(e), r(e)$, as stated in Theorem~\ref{thm: optimal speeds naive wec} are actually global optimizers to $\nlpwbar{e}$.

\begin{lemma}
\label{lem: s(e), r(e) are feasible}
For every $e\in \reals_+$, speeds $s(e), r(e)$, as they are defined in Theorem~\ref{thm: optimal speeds naive wec}, are feasible to $\nlpwbar{e}$.
\end{lemma}

\begin{lemma}
\label{lem: s(e), r(e) are feasible1}
For every $e\in \reals_+$, speeds $s(e), r(e)$, as stated in Theorem~\ref{thm: optimal speeds naive wec}, are the optimal solutions to $\nlpwbar{e}$.
\end{lemma}

\section{Minimizing Evacuation Time with Bounded Linear Makespan Evacuation Energy}
\label{sec: min time, given linear makespan energy}

In this section we study the problem \we{1}{\Delta} of minimizing the makespan evacuation time, given that the  makespan evacuation energy $\Delta=ed$ for some constant $e$. We show how to choose optimal speed values $s,r$ for algorithm \naivew{s,r}. Note that even though $d$ is unknown to the algorithm, speeds $s,r$ may depend on the known value $e$, and the maximum speed $b=1$.

\begin{theorem}
\label{thm: optimal speeds naive wec makespan energy bound}
For every constant $e\in \reals_+$, problem \we{1}{ed} admits a solution by \naivew{s,r}, where speeds $s,r$ are chosen as follows
$$
\begin{tabular}{l|cc}
			& 	$e< 3$ 	&	$e\geq 3$   \\
\hline
$s$ 		& $\sqrt{\tfrac{e}{3}}$		& 	$1$	\\
$r$		&	$\sqrt{\tfrac{e}{3}}$		&	$1$	 \\
\end{tabular}
$$
The induced  evacuation time  is given by $g(e)d$ where 
$$
g(e):=
\left\{
\begin{array}{ll}
3\sqrt{  \tfrac{3}{e}   }							&, e < 3 \\	
1														&, e\geq 3 \\	
\end{array}
\right.
$$
\end{theorem}

\begin{proof} (Theorem~\ref{thm: optimal speeds naive wec makespan energy bound})
What distinguishes the performance, and feasibility, of \naivew{s,r} between \wec{1}{ed} and \we{1}{ed}, is that in the former, the  total evacuation energy (equal to $d(2s^2+2r^2)$) is bounded by $e$, while in the latter the makespan evacuation energy (equal to $d(s^2+2r^2)$) is bounded by $e$. Hence, similar to the analysis for \wec{1}{ed}, the optimal speed choices for \naivew{s,r} to \we{1}{ed} are the optimal solutions to the following NLP. 
\begin{align}
& ~~\min_{s,r \in \reals} \tfrac1s+\tfrac2r \tag{$\nlpwbarbar{e}$}\\
s.t. ~~&
s^2+2r^2 \leq e  \notag \\
&0 \leq s,r\leq 1 \notag
\end{align}
Note that $\nlpwbarbar{e}$ is convex, hence any choice of feasible speeds satisfying 1st order optimality (KKT) conditions is also the unique global minimizer. Moreover, the choices of $s,r$ of the statement of the theorem are clearly feasible to $\nlpwbarbar{e}$. Hence, it suffices to show that the choices of $s,r$ do indeed satisfy KKT conditions. 

When $e<3$ we note that the energy constraint is tight, while both speed constraints are not tight. Hence, $s,r$ are the unique optimizers if there exists $\lambda\geq0$ satisfying 
$$
- \nabla \left( \frac1s+\frac2r\right)
= \lambda \nabla (s^2+2r^2)
\Leftrightarrow
\left(
\begin{array}{c}
1/{s^2}\\
2/{r}^2
\end{array}
\right)
=
\lambda
\left(
\begin{array}{c}
2s\\
4r
\end{array}
\right)
$$
from which we conclude that $\lambda=1/(2s^3)=1/(2r^3)>0$ as wanted (for $s=r=\sqrt{e/3}$). 

When $e\geq 3$ we note that the speed constraints are both tight, while the energy constraint is tight only when $e=3$. In that case, it suffices to show that there exist nonnegative $\lambda_1, \lambda_2$ satisfying

$$
- \nabla \left( \frac1s+\frac2r\right)
= 
\lambda_1
\left(
\begin{array}{c}
1\\
0
\end{array}
\right)
+ 
\lambda_2
\left(
\begin{array}{c}
0\\
1
\end{array}
\right)
$$
Clearly, $\lambda_1=1/s^2=1 > 0$ and $\lambda_2=2/r^2=2> 0$, which concludes the proof. 
\end{proof}

\section{Conclusion} 

We investigated how the wireless  communication model affects time/energy trade-offs for completion of the evacuation task by two robots. Our study raises several interesting problems worth investigating. In addition to improving the trade-offs, it would be interesting to consider search with multiple agents some of which may be faulty in linear~\cite{PODC16,isaacCzyzowiczGKKNOS16} as well as cyclical~\cite{CGGKMP} search domains.

\section*{Acknowledgments}

This research is supported by NSERC discovery grants, NSERC graduate scholarship, and NSF.

\bibliographystyle{plain}
\bibliography{refs}

\appendix
\section{Details of missing proofs from Section~\ref{sec: min energy, given bounds on time, wireless}}

\subsection{Lemma~\ref{No online:lm}}
\begin{proof} (Lemma~\ref{No online:lm})
Fix $0 < \epsilon \leq 3$ and let $bc=3-\epsilon$. We show that no  algorithm can solve problem \ecb{b}{c}. 
For the sake of contradiction, consider a wireless algorithm solving \ecb{(3-\epsilon)/b}{b}, and having  evacuation time no more than $(3-\epsilon)d/b$, if the exit is placed $d$ away from the origin. 
For a large enough $d>0$, we let the algorithm run till the first point among $\pm d$ is reached by a robot (and maybe they are reached simultaneously). Without loss of generality, assume that $+d$ is reached, say by robot $R$, no later than the other point. Note that for this point to be reached, at least time $d/b$ has passed. Now, we place the exit at point $-d$. The additional time that $R$ needs to reach the exit is $2d/b$, for a total time of $3d/b$, a contradiction to the stipulated evacuation time of  $(3-\epsilon)d/b$.
\end{proof}

\subsection{Lemma~\ref{speed-b}}

\begin{proof} (Lemma~\ref{speed-b})
Suppose not. Notice that the points $\pm d$ cannot be visited {\em before } time $d/b$ using speed at most $b$. We look at two cases.

\begin{description}

\item[Case 1: ] There exists $d > 0$ such that neither $d$ nor $-d$ is visited at time $d/b$. Consider the first time $t > d/b$  when either of them is visited, wlog  let the  point $+d$ be visited at time $t > d/b$ by robot $R_1$.  We put the exit at $-d$. Then $R_1$  has to travel an additional distance of $2d$, and can use speed at most $b$, so needs time at least $2d/b$ to get to the exit. The total time taken by $R_1$ to evacuate is at least $t + 2d/b > 3d/b = cd$. 

\item[Case 2:] There exists  $d>0$  such that $d$ is visited at time $d/b$ but $-d$ is not visited at this time (or vice versa). Wlog  suppose $R_1$ is at point $d$ at time $d/b$.  Let $-d+ 2\epsilon$ be the closest point to $-d$ that {\em has} been  visited at time $d/b$ where $\epsilon > 0$ since by assumption $-d$ is not visited at this time.   We put the exit at $-d+\epsilon$.  The time limit to evacuate is $c(d-\epsilon)$. At time $d/b$, $R_1$ is at distance $2d-\epsilon$ from the exit, so the   total time for $R_1$ to reach the exit is at least 
$$d/b + (2d - \epsilon)/b = 3d/b - \epsilon/b = cd - \frac{c \epsilon}{3} > cd - c \epsilon$$

\end{description}

In both cases, we showed that the robots cannot evacuate in the required time bound. This completes the proof by contradiction.  
\end{proof}

\subsection{Lemma~\ref{lem: naive as NLP wireless}}

\begin{proof} (Lemma~\ref{lem: naive as NLP wireless})
For fixed parameter $b$, consider NonLinear Program 
\begin{align}
& ~~\min_{s,r \in \reals} ~~\frac1s+\frac2r \\
s.t. ~~&0 \leq s,r\leq b \notag
\end{align}

$\ttimew{s,r}=\frac1s+\frac2r$ is clearly strictly decreasing in any of $s,r>0$, hence in an optimizer both constraints $s,r\leq b$ have to be tight, for any fixed $b$. But then, $\min_{s,r \in \reals} \ttimew{s,r} = \ttimew{b,b}=3/b$.
It follows that $\nlpw{c}{b}$ has a feasible solution if and only if $3/b\leq c$.

By Lemma~\ref{lem: energy time formulas wireless}, it is immediate that $\nlpw{c}{b}$ exactly models the problem of choosing optimal speeds for \naivew{s,r}, for problem \ecb{b}{c}. 
Also note that $\ttimew{s,r}$ and $\eenergyw{s,r}$ are strictly convex functions when $s,r>0$,
and hence $\nlpw{c}{b}$ is a convex program. Moreover, an optimizer always exists, since the function is bounded from below, and is defined over a compact feasible region. 
Finally, the claim pertaining to the competitive ratio follows from Lemma~\ref{lem: energy time formulas wireless}. 
\end{proof}

\subsection{Lemma~\ref{lem: opt wireless 3<=cb<=3.25}}

\begin{proof} (Lemma~\ref{lem: opt wireless 3<=cb<=3.25})
1st order optimality (KKT) conditions for $s,r$, assuming that time constraint and $r\leq b$ are tight, are
\begin{align*}
-\nabla \eenergyw{s,r} & = 
\lambda_1
\nabla \ttimew{s,r}
+
\lambda_2 
\left(
\begin{array}{c}
0\\
1
\end{array}
\right)  \\
\ttimew{s,r}&=c\\
0 \leq s &\leq b\\
r&=b \\
\lambda_1,\lambda_2&\geq 0,
\end{align*}
where functions $\eenergyw{\cdot,\cdot}, \ttimew{\cdot,\cdot}$ are as in \eqref{equa: ttimew}, \eqref{equa: eenergyw}.
Utilizing only the equality constraints above, we easily derive that 
$$
s=\frac{b}{b c-2}, ~~
\lambda_1=\frac{2 b^3}{(b c-2)^3}, ~~
\lambda_2=-\frac{2 \left(b^4 c^3-6 b^3 c^2+12 b^2 c-10 b\right)}{(b c-2)^3}
$$
Note that $s\leq b$ for all $cb\geq 3$. $\lambda_1$ is clearly positive. It is enough to verify that $\lambda_2\geq 0$. 

Indeed, define $g(x)=10 - 12 x + 6 x^2 - x^3$, and note that $\lambda_2 = b \frac{g(cb)}{(cb-2)^3}$. Since $cb\geq 3$ and $b>0$, we conclude that $\lambda_2\geq0$ as long as $g(cb)\geq 0$. For that, we calculate the 3 roots of $g$
$$
2+\sqrt[3]{2}, ~~2-\frac{1\pm i \sqrt{3}}{2^{2/3}}.
$$
Since the leading coefficient of $g$ is negative, and since $g$ has a unique real root, we conclude that $g(x)\geq 0$ as long as $x\leq 2+\sqrt[3]{2}$ as wanted. 
\end{proof}

\subsection{Lemma~\ref{lem: opt wireless cb>3.25}}

\begin{proof} (Lemma~\ref{lem: opt wireless cb>3.25})
1st order optimality (KKT) conditions for $s,r$, assuming that only time constraint is tight, are
\begin{align*}
-\nabla \eenergyw{s,r} & = 
\lambda
\nabla \ttimew{s,r}
 \\
\ttimew{s,r}&=c \\
0\leq s,r&\leq b \\
\lambda &\geq 0
\end{align*}
Using only the equality constraints, we derive 
$$
s=\frac{1+2^{2/3}}{c}, ~~
r=\frac{1}{2} s (c s-1)^2=\frac{2+\sqrt[3]{2}}{c}, ~~
\lambda=2s^3.
$$
Observe that the proposed values of $s,r$ satisfy the speed bound only if $cb \geq \delta$. But then, we also see that $\lambda>0$ for all such $c,b$, and hence $s,r$ do indeed satisfy the 1st order optimality conditions. 
\end{proof}

\section{Details of missing proofs from Section~\ref{sec: min time, given linear energy}}

\subsection{Lemma~\ref{lem: s(e), r(e) are feasible}}
\begin{proof} (Lemma~\ref{lem: s(e), r(e) are feasible})
Speeds $s=s(e)$ and $r=r(e)$ are clearly non negative. Next we verify that they never attain value more than 1. We examine two cases. When $e< 2+\sqrt[3]{2}$, it is easy to see that $r/s=\sqrt[3]{2}$. Hence, it is enough to check that $r\leq 1$, which is immediate from the formula of $r=r(e)$. 
In the other case, we assume $e\in [\delta,4)$. Speed $r$ is clearly at most 1, as well as $s=\sqrt{(e-2)/2}\leq \sqrt{(4-2)/2}=1$. 

Next we verify that the given speeds comply with the evacuation energy bounds. When $e< 2+\sqrt[3]{2}$ we have 
$$
2(s^2+r^2)
=
2e
\left(
\frac{1}{2\left(1+2^{2/3}\right)}
+
\frac{1}{\left(2+2^{1/3}\right)}
\right)
=e.
$$
When $e\in [\delta,4)$ we have 
$$
2(s^2+r^2)
=
2e
\left(
\frac{e-2}{2}
+
1
\right)
=e.
$$
Lastly, when $e\geq 4$ both speeds are 1, and clearly, $2(s^2+r^2)= 4 \leq e$ as wanted. 
\end{proof}

\subsection{Lemma~\ref{lem: s(e), r(e) are feasible1}}

\begin{proof} (Lemma~\ref{lem: s(e), r(e) are feasible1})
For every $e\in \reals_+$, 
we verify that speeds $s(e), r(e)$ satisfy 1st order optimality conditions. Since $\nlpwbar{e}$ is convex, that would imply that $s(e), r(e)$ are the unique optimizers. 

First we observe  that the energy inequality constraint is always tight (verified within the proof of Lemma~\ref{lem: s(e), r(e) are feasible}). 
Apart from that constraint, let $I(e)$ (possibly empty) denote the set of constraints, among $s,r\leq 1$, which are tight for the specific candidate optimizer $s(e), r(e)$, and for a specific value of $e$. For $i\in I$ we denote the corresponding constraint by $g_i(s,r)\leq 1$. 

When $e< 2+\sqrt[3]{2}$, the bound constraint is the only constraint which is tight. Therefore KKT conditions are satisfied as long as there exists $\lambda\geq 0$ such that 
$$
- \nabla \left( \frac1s+\frac2r\right)
= \lambda \nabla 2(s^2+r^2)
\Leftrightarrow
\left(
\begin{array}{c}
1/{s^2}\\
2/{r}^2
\end{array}
\right)
=
\lambda
\left(
\begin{array}{c}
4s\\
4r
\end{array}
\right)
$$
A solution exists as long as 
$
2s^3=r^3
$, which is indeed, the case, 
which also implies that $\lambda = 1/(4s^3)\geq 0$.

When $e\in [2+\sqrt[3]{2}, 4)$, the bound constraint and constraint $r\leq 1$ are tight. Therefore KKT conditions are satisfied as long as there exist $\lambda_1, \lambda_2 \geq 0$ such that 
$$
- \nabla \left( \frac1s+\frac2r\right)
= \lambda_1 \nabla 2(s^2+r^2)
+ 
\lambda_2
\left(
\begin{array}{c}
0\\
1
\end{array}
\right)
\Leftrightarrow
\left(
\begin{array}{c}
1/{s^2}\\
2/{r}^2
\end{array}
\right)
=
\lambda_1
\left(
\begin{array}{c}
4s\\
4r
\end{array}
\right)
+ 
\lambda_2
\left(
\begin{array}{c}
0\\
1
\end{array}
\right)
$$
Solving for $\lambda_1, \lambda_2$, and using the provided values for $s=s(e)$ and $r=r(e)$ we obtain 
$$
\lambda_1 = \sqrt{2}
\left(\frac{ 1+2^{2/3}}{e}\right)^{3/2}, ~~
\lambda_2 = 
2-2\left(\frac{2+\sqrt[3]{2}}{e}\right)^{3/2}, 
$$
and clearly both values are nonnegative when $e\geq 2+\sqrt[3]{2}$. 

Lastly, for the 1st order optimality conditions, 
when $e\geq 1$, all (but the non-negativity constraints) are tight. Therefore KKT conditions are satisfied as long as there exist $\lambda_1, \lambda_2, \lambda_3 \geq 0$ such that 
\begin{align*}
&
- \nabla \left( \frac1s+\frac2r\right)
= \lambda_1 \nabla 2(s^2+r^2)
+ 
\lambda_2
\left(
\begin{array}{c}
1\\
0
\end{array}
\right)
+ 
\lambda_3
\left(
\begin{array}{c}
0\\
1
\end{array}
\right)\\
&\Leftrightarrow
\left(
\begin{array}{c}
1/{s^2}\\
2/{r}^2
\end{array}
\right)
=
\lambda_1
\left(
\begin{array}{c}
4s\\
4r
\end{array}
\right)
+ 
\lambda_2
\left(
\begin{array}{c}
1\\
0
\end{array}
\right)
+ 
\lambda_3
\left(
\begin{array}{c}
0\\
1
\end{array}
\right)
\end{align*}
Since $r=s=1$, the above system simplifies to 
\begin{align*}
1&=4\lambda_1+\lambda_2\\
2&=4\lambda_1+\lambda_3
\end{align*}
which admits the solution $\lambda_1=1/4\geq0, \lambda_2=0, \lambda_3=1\geq0$. 
\end{proof}

\end{document}